\documentclass{article}

\usepackage{amsmath}
\usepackage{amssymb}
\usepackage{theorem}
\usepackage{enumerate}
\usepackage{microtype}
\usepackage{framed}
\usepackage{hyperref}

\usepackage{pgf,tikz}
\usetikzlibrary{shapes,arrows,automata,positioning,backgrounds,fit,decorations.pathreplacing}

\newtheorem{definition}{Definition}
\newtheorem{theorem}{Theorem}
\newtheorem{proposition}[theorem]{Proposition}
\newenvironment{proof}[1][Proof]{\paragraph{{#1}}}%
                {{\hfill\(\Box\)\\}}
\newcommand{\cand}{\text{ and }}

\newcommand{\bra}{\langle}
\newcommand{\ket}{\rangle}

\newcommand{\coll}[1]{\ensuremath{\left\{ {#1}\right\} }}

\newcommand{\fall}[1]{{\forall\,{#1},\ }}

\newcommand{\mc}[1]{{\mathcal{#1}}}

\newcommand{\mb}[1]{{\bf #1}}

\newcommand{\Mes}{\mathop{\mathrm{Mes}}}
\newcommand{\Uni}{\mathop{\mathrm{Uni}}}
\newcommand{\appl}{\mathop{\mathrm{appl}}\nolimits}

\newcommand{\blank}{\,\rule{5pt}{.5pt}\,}

\begin{document}

\title{Weakening the Born Rule -- \\ Towards a Stateless Formulation of Quantum Mechanics}
\author{Olivier Brunet \\
\texttt{ol{}ivier.{}brunet\ \hbox{at}\ normalesup.org}}

\maketitle

\begin{abstract}
The notion of state vector is, in quantum mechanics, as central as it is problematic, as illustrates the wealth of publications about the subjects, including in particular the many attempts to obtain an acceptable interpretation of quantum mechanics.

In this article, we propose a different approach, and initiate the study of a formulation of quantum mechanics, where the notion of state is entirely replaced by assertions about measurement outcomes. We define a notion of ``verification'' which represents the knowledge that one may have about the possible outcomes of the measurements performed on a quantum system, and express a set of logical rules which allow to reason about quantum systems using verification assertions only, and thus making no reference to the problematic notion of state. 
\end{abstract}

\section{Introduction}

Central to the mathematical formulation of quantum mechanics is the notion of \emph{state vector}. As one can read in any textbook, every physical system $S$ can be associated to a Hilbert space $\mc H$ such that states of $S$ are represented by rays (one-dimensional subspaces) of $\mc H$. Using state vectors, one can easily describe the evolution of a quantum system:
\begin{itemize}
  \item in a reversible way, which can be decribed by a unitary operator $U$ acting on $\mc H$. In that case, a system in state $|\varphi\ket$ is transformed into a system in state $U |\varphi\ket$;
  \item in a non-reversible way, by performing a \emph{measurement}: if a system in state $|\varphi\ket$ is measured with observable $\mc O$ (which is represented by an hermitian operator $O$), then the outcome is an eigenvalue $a$ of $O$, and the system's state becomes a vector belonging to the eigenspace $E_a$ associated to $a$, namely the orthogonal projection of $|\varphi\ket$ on $E_a$. It is worth noting that this implies that $E_a$ cannot be orthogonal to $|\varphi\ket$, which means that when measuring $S$, some outcomes are impossible to obtain.
\end{itemize}

However, almost a century after the elaboration of this formalism, the very nature of state vectors remains extremely problematic. Among important questions is whether the state vector does represent reality or simply our knowledge of it \cite{Spekkens2010Epistemic}. If the state vector does indeed directly represent reality, is it a complete representation or can it be complemented by some supplemental \emph{hidden variables}? These interrogations can be traced back to the very beginning of quantum mechanics with Einstein's criticisms (in the famous EPR article \cite{Einstein35EPR}, and even before at the 1927 Solvay conference \cite{QuantumCrossroads}). Other important questions are, among others, whether quantum mechanics is intrinsically probabilistic and whether it can comply with locality. However many important results, such as for instance no-go theorems (including Bell's theorem \cite{Bell64,Mermin93}, Kochen and Specker's theorem \cite{KochenSpecker67,Mermin93,Brunet07PLA} and more recent results \cite{PRB2012,Brunet13Dynamics}) have not completely succeeded to lift the `unease' regarding quantum mechanics, as Rovelli puts it in \cite{Rovelli96RQM}.

\ 

Yet, this notion of state vector is a purely abstract construction and cannot be ``accessed'' directly. Indeed, the only way to obtain actual data about a quantum system is by \emph{measuring} it. To quote Rovelli \cite{Rovelli96RQM} again,

\begin{quotation}
\emph{``the physical content of the theory is given by the outcomes of the measurements.''}
\end{quotation}

and, later

\begin{quotation}
\emph{``anything in between two measurement outcomes is like the “non-existing” trajectory of the electron, to use Heisenberg’s vivid expression, of which there is nothing to say.''}
\end{quotation}

As a consequence, rather than defining quantum mechanics using state vectors as the basic notion, it seems that it would be beneficial to attempt to reverse the perspective and to define a measurement-based formulation of quantum mechanics (or, at least, a large fragment of it), with absolutely no reference whatsoever to state vectors.

The expected benefits of such an approach are the obtention of a purely epistemic formulation of quantum mechanics which would rest on experimental data only and avoid the use (at least, at first) of abstract and interpretational elements and thus would not refer in any way to some problematic ``real state of affairs''. 

\ 

The present article describes the first elements of such an attempt, in which we will restrict ourselves to nondegenerate observables.

\section{Quantum Measurement Logic}

\subsection{Verification Judgements}

Quantum mechanics teaches us that to every closed quantum system $S$, one can associate a Hilbert space $\mc H$ such that every observables for $S$ can be represented by a unitary operator acting on $\mc H$.

Stated this way, the use of Hilbert spaces is not motivated by the existence of a state vector which would supposedly belong to it. Instead, it follows from the fact that all the possible outcomes of observables which can be applied to a given quantum system can conveniently be represented as a Hilbert space or, more precisely, as the closed subspaces of a Hilbert space (we identify here an outcome of $\mc O$ with an eigenspace of the corresponding hermitian operator $O$ rather than with the associated eigenvalue). Thus, the use of Hilbert spaces is entirely justified by the consideration of actual experimental data regarding $S$ rather than from abstract assumptions.

Following this idea, throughout this article, we will identify every observable $\mc O$ with the set of its outcomes, which we view as the eigenspaces of the corresponding hermitian operator $O$. Moreover, as stated earlier, we will only consider non-degenerate observables, so that each eigenspace is one-dimensional. As a consequence, we can represent the possible outcomes of a non-degenerate observable as an orthonormal basis of $\mc H$ and we will hereafter represent observables of a quantum system $S$ modelled by Hilbert space $\mc H$ as orthonormal basis of $\mc H$.

It can be remarked that such a representation is unique, up to phase factors.

\ 

Since we chose to base our approach on measurement outcomes, let us first focus on the \emph{projection postulate} which we quote here from \cite{Rae2002QM}:
\begin{quotation}
\noindent \textbf{Postulate 4.2} {[}...{]} 
Immediately after such a measurement, the wavefunction of the system will be identical to the eigenfunction corresponding to the eigenvalue obtained as a result of the measurement.
\end{quotation}
%
%
Before studying the way this postulate can be expressed without making any reference to state vectors, let us first remark that one has to be able to differentiate a quantum system before and after it has been measured.
In order to do this simply, we suggest the following notation:

\begin{definition}
If the measurement of a quantum system $S$ with observable $O$ yields outcome $|\varphi\ket$, we denote this by writing
$$ (S', |\varphi\ket) = \Mes(S, \mc O). $$
Moreover, with this notation, $S'$ represents the system right after the measurement occured. 
\end{definition}

With this notation, let us first express the fact that we consider that every outcome of observable $\mc O$ is an element of $\mc O$ (which, we recall, is seen as an orthonormal basis of $\mc H$) by the following implication:
$$ (S', |\varphi\ket) = \Mes(S, \mc O) \implies |\varphi\ket \in \mc O. $$
In that expression, the specifications of the system before (that is, $S$) and after (that is $S'$) the measurement play no role. In order to have a lighter expression, we write this~as
$$ (\blank, |\varphi\ket) = \Mes(\blank, \mc O) \implies |\varphi\ket \in \mc O, $$
where, by writing ``$\blank$'', we mean that the corresponding element is not relevant in the considered assertion (even though it has some physical reality).

\ 

Returning to the \emph{projection postulate}, this postulate, used together with the Born Rule, tells that if a quantum system $S$ has been measured with outcome $|\varphi\ket$ and if it measured again immediately after the first measurement, then any outcome orthogonal to $|\varphi\ket$ is impossible. This can be written, using appropriate ``$\blank$'', as follows:
$$ (S, |\varphi\ket) = \Mes (\blank, \blank) \implies \fall {(\blank, |\psi\ket) = \Mes(S, \blank)} \bra \varphi | \psi \ket \neq 0. $$

Let us now consider the unitary evolution of a closed system. If the application to system $S$ of such a transformation represented by unitary operator $U$ yields $S'$, we write
$$ S' = \Uni(S,U). $$
Regarding measurement outcomes, if $S$ is known to be such that any measurement outcome orthogonal to a vector $|\varphi\ket$ is impossible (such a statement can follow from the application of the projection postulate), then it is clear that for $S' = \Uni(S,U)$, any measurement outcome orthogonal to $U|\varphi\ket$ is impossible for $S'$, and reciprocally. This can be expressed, formally,~as
\begin{multline}
S' = \Uni(S, U) \implies \\
\Bigl[ \fall{(\blank, |\psi\ket) = \Mes(S, \blank)} \bra \varphi | \psi \ket \neq 0 \iff \\ \fall{(\blank, |\psi\ket) = \Mes(S', \blank)} \bra \varphi | U | \psi \ket \neq 0 \Bigr].
\end{multline}

\ 

In the formal translations of the projection postulate and the unitary evolution of a system, one can notice that a particular pattern appears on several occurences, of the~form
$$\hbox{“}~\fall{(\blank, |\psi\ket) = \Mes(S, \blank)} \bra \varphi | \psi \ket \neq 0.~\hbox{”}$$
This suggests the definition of the following notation:
\begin{definition}
Given a system $S$ represented by a Hilbert space $\mc H$ and a vector $|\varphi\ket \in \mc H$, let the logical judgement $S \models |\varphi\ket$ denote the statement
$$ \fall{(\blank, |\psi\ket) = \Mes(S, \blank)} \bra \varphi | \psi \ket \neq 0.$$
We then say that system $S$ \emph{verifies} $|\varphi\ket$.
\end{definition}

Using this notation, the previous rules can be reexpressed~as:
\begin{description}
  \item[Projection Postulate]
  $$ (S, |\varphi\ket) = \Mes(\blank, \blank) \implies S \models |\varphi\ket $$
  \item[Unitary Evolution] 
  $$ S' = \Uni(S, U) \implies \bigl(S \models |\varphi\ket \iff S' \models U |\varphi\ket\bigr) $$
\end{description}

Before continuing with the definition of our formalism for taking composite systems into account, let us first make a few comments about this notion of verification. Despite an apparent similarity with state vectors, (since in each case, we relate a quantum system and a vector belonging to the corresponding Hilbert space), both notions are extremely different.

The first obvious difference is that state vectors are abstract elements, which cannot be accessed directly. Meanwhile, a verification judgement is a statement which relates exclusively to actually obtainable experimental data. As a consequence, the statement $S \models |\varphi\ket$ can be directly refuted.

Moreover, while it is postulated that a closed quantum system is always in some definite state vector
, it makes no sense to state in general that, given a quantum system $S$, there exists a vector $|\varphi\ket$ such that $S \models |\varphi\ket$. Indeed, a verification judgement can only be obtained using the projection postulate rule, or at least deduced using the previous rules (together with a few more rules which will be presented next).

Another major difference is that, as we will see, it is possible to express several distinct yet valid verification judgements about a single quantum system, which contrasts deeply with the supposed (and problematic) uniqueness of a state vector.

At last important point that would be made is that the validity of a verification judgement $S \models |\varphi\ket$ is independent of whether $S$ has actually been measured. Indeed, if $S$ has not been measured yet, such a judgement can be seen as predictive. If $S$ has been measured, such a judgement can be seen as an assessement of the validity of quantum mechanics. Even more importantly, this independence regarding whether the measurement has actually happened is particularly welcome when the considered measurement is supposed to take place at a space-like separated location. 



\ 



From the definition of a verification judgement, it is clear that the following rule is valid:
\begin{equation}
\bigl[S \models |\varphi\ket \cand (\blank, |\psi\ket) = \Mes(S, \blank)\bigr] \implies \bra \varphi | \psi \ket \neq 0.
\end{equation}
We call it the \emph{Weak Born Rule}, since it can be seen as a weakening of the usual Born Rule. Indeed, the usual Born Rule states that if a system $S$ is in a normalized state $|\varphi\ket$, the probability of obtaining an outcome corresponding to normalized eigenvector $|\psi\ket$ equals $\bigl|\bra \varphi \vert \psi \ket \bigr|^2$ when measuring $S$ with a nondegenerate observable. This result does in particular imply that if $|\psi\ket$ is orthogonal to $|\varphi\ket$, it cannot be obtained as an outcome. Equivalently, for an eigenvector $|\psi\ket$ to be possible as an outcome, it must not be orthogonal to $|\varphi\ket$. The Weak Born Rule is precisely expressing the equivalent result, where the statement ``$S$ is in state $|\varphi\ket$'' is replaced by ``$S \models |\varphi\ket$''.

Actually, this is just an other way to express the definition of verification judgements which are the main component of our approach, and this remark motivates the title of this article.

One should be aware that all this rule indicates is whether a given outcome is possible or not. However, stating that an outcome has to be considered as possible (that is, has not been ruled out by the Weak Born Rule) does not imply anything regarding whether it \emph{should\/} be obtainable somehow.

A notable exception occurs when $S$ verifies $|\varphi\ket$ and is measured with an (non degenerate) observable $\mc O$ compatible with $|\varphi\ket$ (which means that $|\varphi\ket$ is an eigenvector of $\mc O$ or, in our approach, that $|\varphi\ket$ belongs to $\mc O$). If the measurement of $S$ with observable $\mc O$ yields outcome $|\psi\ket$, then $|\psi\ket$ belongs to $\mc O$, and because of the \emph{Weak Born Rule}, it cannot be orthogonal to $|\varphi\ket$. The only suitable element of $\mc O$ is $|\varphi\ket$, so that it is the only possible outcome:
$$ \bigl(S \models |\varphi\ket \cand |\varphi\ket \in \mc O \cand (\blank, |\psi\ket) = \Mes(S, \mc O)\bigr) \implies |\varphi\ket = |\psi\ket. $$




\subsection{Composite systems}

\paragraph{Tensor Products} The quantum systems considered so far were all supposed to verify an important implicit property: if a system $S$ is modelled by Hilbert space $\mc H$, then \emph{any} orthonormal basis of $\mc H$ can give rise an observable. A quantum system which verifies this property will be called ``simple''. 


Now, given $n$ quantum systems $S_1, \ldots, S_n$, respectively modelled by Hilbert spaces $\mc H_1, \ldots, \mc H_n$, one can form a composite system $S = (S_1, \ldots, S_n)$, modelled by the tensor product $\mc H = \mc H_1 \otimes \cdots \otimes \mc H_n$. This composite system, contrary to simple systems, does not verify the fact that any orthonormal basis of $\mc H$ corresponds to an observable of $S$.

Indeed, the measurement of $S$ does in fact reduce to the measurement of each of subsystem $S_1, \ldots, S_n$ and the outcome $|\varphi_S\ket$ is then the tensor product of outcomes $|\varphi_1\ket, \ldots, |\varphi_n\ket$:
$$ |\varphi_S\ket = |\varphi_1\ket \otimes \cdots \otimes |\varphi_n\ket.$$
The latter is the general form for outcomes of $S$. 
This situation can be illustrated by the fact that in order to measure a system made of two entangled particles, one has to measure each particle separatly and then gather both outcomes. 

\ 

This remark can be used to prove a first simple result on verification judgements regarding composite systems:

\begin{proposition}
$$ (S_1, S_2) \models |\varphi_1 \otimes \varphi_2\ket \iff \bigl( S_1 \models |\varphi_1\ket \cand S_2 \models |\varphi_2\ket\bigr) $$
\end{proposition}
\begin{proof}
For all $(|\psi_1\ket, |\psi_2\ket) \in \mc H_1 \times \mc H_2$, one~has
$$ \bra \varphi_1 \otimes \varphi_2 \vert \psi_1 \otimes \psi_2 \ket = \bra \varphi_1 \vert \psi_1 \ket \bra \varphi_2 \vert \psi_2 \ket.$$
As a consequence, 
\begin{align*}
(S_1, S_2) \models |\varphi_1 \otimes \varphi_2\ket & \iff \fall {(|\psi_1\ket, |\psi_2\ket) \in \mc H_1 \times \mc H_2} \bra \varphi_1 \otimes \varphi_2 \vert \psi_1 \otimes \psi_2 \ket \neq 0 \\
& \iff \fall {(|\psi_1\ket, |\psi_2\ket) \in \mc H_1 \times \mc H_2} \bra \varphi_1 \vert \psi_1 \ket \bra \varphi_2 \vert \psi_2 \ket \neq 0 \\
& \iff \fall {|\psi_1\ket \in \mc H_1} \bra \varphi_1 \vert \psi_1 \ket \neq 0 \cand \\ & \phantom{\iff} \qquad \qquad \fall {|\psi_2\ket \in \mc H_2} \bra \varphi_2 \vert \psi_2 \ket \neq 0 \\
& \iff S_1 \models |\varphi_1\ket \cand S_2 \models |\varphi_2\ket.
\end{align*}
\end{proof}
\paragraph{Partial Measurements}

Suppose that that two quantum system $S_1$ and $S_2$ (modelled respectively by $\mc H_1$ and $\mc H_2$) are known to verify $|\psi\ket \in \mc H_1 \otimes \mc H_2$:
\begin{equation*}
  (S_1, S_2) \models \mathopen|\psi\ket.
\end{equation*}
Suppose moreover that $S_1$ has been measured, yielding outcome $|\varphi_1\ket$:
$$ (\blank, |\varphi_1\ket) = \Mes(S_1, \blank). $$
An interesting question to consider is the determination of whether the knowledge of outcome $|\varphi_1\ket$ together with that of $(S_1, S_2) \models |\psi\ket$, does provide extra knowledge regarding $S_2$ (in the form, of course, of a verification judgement).

In order to answer this question positively, let $|\varphi_2\ket$ denote the potential outcome of a measurement performed on $S_2$. From the statement $(S_1, S_2) \models |\psi\ket$, one should~have
\begin{equation}
  \bra \psi \vert \varphi_1 \otimes \varphi_2 \ket \neq 0. \label{eq:appl1}
\end{equation}
Considering an orthogonal basis $\coll{|e_i\ket}$ (resp. $\coll{|f_j\ket}$) of $\mc H_1$ (resp. $\mc H_2$), let us~write
$$ |\varphi_1\ket = \sum_i \alpha_i |e_i\ket, \quad |\varphi_2\ket = \sum_j \beta_j |f_j\ket \quad \hbox{and} \quad |\psi\ket = \sum_{i,j} \gamma_{i,j} |e_i \otimes f_j\ket.$$
Equation \ref{eq:appl1} is then equivalent~to
$$ \sum_{i,j} \gamma_{i,j}^\star \alpha_i \beta_j \neq 0, $$
which can also be expressed~as
\begin{equation}
  \sum_j \Bigl(\sum_i \gamma_{i,j} \alpha_i^\star \Bigr)^{\!\star} \beta_j \neq 0.
\end{equation}
This equality can be seen as expressing the inner product between $|\varphi_2\ket$ and another vector of $\mc H_2$ which depends on $|\varphi_1\ket$ and $|\psi\ket$, leading to the following definition.
\begin{definition}
Using the previous notations, given two vectors $|\varphi_1\ket \in \mc H_1$ and $|\psi\ket \in \mc H_1 \otimes \mc H_2$, we define the vector $\mathopen| \appl_1(\varphi_1, \psi) \ket$ of $\mc H_2$~as
\begin{equation}
\mathopen| \appl_1(\varphi_1, \psi) \ket = \sum_j \Bigl(\sum_i \gamma_{i,j} \alpha_i^\star \Bigr) |f_j\ket. \label{eq:appl2}
\end{equation}
\end{definition}

We call the  anti-linear application $\mathopen|\varphi_1\ket \mapsto \mathopen| \appl_1(\varphi_1, \psi) \ket$ (or, written in a shorter form, as $\mathopen| \appl_1(\ \cdot\ , \psi) \ket$) the \emph{partial application} of $|\psi\ket$ at place $1$. Similarly, one can define, for $|\psi\ket \in \mc H_1 \otimes \cdots \otimes \mc H_n$, the partial application $\mathopen|\appl_k(\ \cdot\ , \psi)\ket$ for any $k$ ranging from $1$ to $n$.

The name ``partial application'' refers to the fact that, regarding the pos\-si\-bi\-lity of outcomes, a vector $|\psi\ket \in \mc H_1 \otimes \cdots \otimes \mc H_n$ is can be seen as an application
\begin{align*}
\mc H_1 \times \cdots \times \mc H_n & \rightarrow \mb C \\
(|\varphi_1\ket, \ldots, |\varphi_n\ket) & \mapsto \bra \psi \vert \varphi_1 \otimes \cdots \otimes \varphi_n \ket.
\end{align*}
To that respect, the specification of one of the $n$ arguments does indeed correspond to the partial application of the previous function.

\ 

We thus have $\bra \psi \,|\, \varphi_1 \otimes \varphi_2 \ket = \bra \appl_1(\varphi_1, \psi) \mathbin| \varphi_2\ket$ so that, returning to our initial question, if $(S_1, S_2) \models |\psi\ket$ and if $(\blank, |\varphi_1\ket) = \Mes(S_1, \blank)$, then any outcome $|\varphi_2\ket$ obtained by measuring $S_2$ should verify $\bra \psi \,|\, \varphi_1 \otimes \varphi_2\ket \neq 0$ or equivalently $\bra \appl_1(\varphi_1, \psi) | \varphi_2\ket \neq 0$, which can be state~as
\begin{multline}
\bigl[(S_1, S_2) \models |\psi\ket \hbox{\ and\ } (\blank, |\varphi_1\ket) = \Mes(S_1, \blank) \bigr] \implies \\ \fall {(\blank, |\varphi_2\ket) = \Mes(S_2, \blank)} \bra \appl_1(\varphi_1, \psi) | \varphi_2\ket \neq 0.
\end{multline}
But one can recognize on the right-hand side of this implication a verification judgement. Thus, we have proved:
\begin{proposition}
Given two quantum system $S_1$ and $S_2$, one~has
\begin{equation}
  \bigl[(S_1, S_2) \models |\psi\ket \hbox{\ and\ } (\blank, |\varphi_1\ket) = \Mes(S_1, \blank) \bigr] \implies S_2 \models \mathopen| \appl_1(\varphi_1, \psi) \ket.
\end{equation}
\end{proposition}


We insist on the fact that such a deduction of $S_2 \models \mathopen| \appl_1(\varphi_1, \psi) \ket$ from both $(S_1, S_2) \models |\psi\ket$ and $(\blank, |\varphi_1\ket) = \Mes(S_1, \blank)$ consists only in providing some new 
knowledge about the system. It only relates to measurement outcomes and does not imply anything about something that would or should happen to the considered system.



\subsection{Summary}

All the rules previously obtained are summarized in figure \ref{fig:hilbert_postulates} and constitute what we call \emph{Quantum Measurement Logic}, or QML. It consists mainly in two parts:
\begin{enumerate}
  \item the description of an experimental setup, with two different types of elements: unitary operators and measurements;
  \item statements about possible measurement outcomes. As stated earlier, they are of the form ``if $S$ is measured, then the outcome cannot be orthogonal of such vector''. 
\end{enumerate}

This approach provides a general framework for reasoning about quantum system. However, it is extremely different from usual quantum logic approaches (\cite{Birkhoff36QuantumLogic,Svozil98Book,DallaChiara2001QuantumLogic}) where the logical inferences are intended to be made ``from within'' the Hilbert lattice. In our approach, the basic elements and the rules are all part of classical 
logic, and the Hilbert lattices are only present as the set to which measurements outcomes belong.

From the design of this logic, it is clear that it is correct with regards to quantum mechanices, so that any statement which can be proved with QML is a valid statement of quantum mechanics. 

\begin{figure}
\begin{framed}
\textbf{Outcome Definition}
\[
    (\blank, |\varphi\ket) = \Mes(\blank, \mc O) \implies |\varphi\ket \in \mc O
\]
\textbf{Weak Born Rule}
\[ \bigl[S \models |\varphi\ket \cand (\blank, |\psi\ket) = \Mes(S, \blank)\bigr] \implies  \bra \varphi | \psi \ket \neq 0 \]
\textbf{Projection Postulate}
\[ (S, |\varphi\ket) = \Mes(\blank, \blank) \implies S \models |\varphi\ket \]
\textbf{Unitary Evolution}
\[ S' = \Uni(S, U) \implies \bigl[ S \models |\varphi\ket \iff  S' \models U |\varphi\ket \bigr] \]
\textbf{Tensor product}
\[ (S_1, S_2) \models |\varphi_1 \otimes \varphi_2\ket \iff \bigl(S_1 \models |\varphi_1\ket \cand S_2 \models |\varphi_2\ket \bigr) \]
\textbf{Partial Measurement}
\[ \bigr[(S_1, S_2) \models |\psi\ket \cand (\blank, |\varphi_1\ket) = \Mes(S_1, \blank) \bigr] \implies S_2 \models \mathopen|\appl_1(\varphi_1, \psi)\ket \]
\caption{Rules of Quantum Measurement Logic}
\label{fig:hilbert_postulates}
\end{framed}
\end{figure}

\section{Illustration}

We now illustrate the use of QML by developping three examples. The first one, the teleportation circuit, shows that QML is expressive enough to capture some caracteristic elements of quantum mechanics, such as entanglement. The second example will focus on the measurement of an EPR pair. Finally, we will turn to Hardy's paradox, and show that a careful examination of this paradox leads to the conclusion that it cannot be expressed in QML and thus relies on supplemental and interpretational elements.


\subsection{Quantum Teleportation}

We consider the quantum teleportation scheme \cite{Bennett93Teleportation}. The first part of the corresponding setup can be represented in terms of a quantum circuit as in figure \ref{fig:quantum_teleportation}.

\begin{figure}
  \begin{framed}
\hbox{} \hfill \begin{tikzpicture}

\tikzstyle{operator} = [draw,fill=white,minimum size=1.5em] 
\tikzstyle{phase} = [fill,shape=circle,minimum size=5pt,inner sep=0pt]
\tikzstyle{point} = [fill,shape=circle,minimum size=3pt,inner sep=0pt]
\tikzstyle{surround} = [fill=blue!10,thick,draw=black,rounded corners=2mm]

\matrix [column sep = {8mm,between origins}, row sep = {8mm,between origins}] {
  \node (aa) [label = left:$\alpha |0\ket + \beta |1\ket$] {} ; &
  \node[point] (a1) [label = above:$C_1$] {} ; &
  \node[phase] (ab) {} ; &
  \node[point] (a2) [label = above:$C_2$] {} ; &
  \node[operator] (ac) {H} ; &
  \node[point] (a3) [label = above:$C_3$] {} ; &
  \node (ad) {} ; \\
  \node (ba) {} ; &
  \node[point] (b1) [label=above:$A_1$] {} ; &
  \node[inner sep = 0pt, outer sep = 0pt] (bb) {$\bigoplus$} ; &
  \node[point] (b2) [label=above:$A_2$] {} ; & & &
  \node (bd) {} ; \\
  \node (ca) {} ; &
  \node[point] (c1) [label=above:$B_1$] {} ; & & & & & 
  \node (cd) {} ; \\
} ;

\draw [-] (aa) -- (a1) -- (ab) -- (a2) -- (ac) -- (a3) -- (ad) ;
\draw [-] (ba) -- (b1) -- (bb) -- (bd) ;
\draw [-] (ca) -- (c1) -- (cd) ;
\draw [-] (bb) -- (ab) ;

\draw[decorate,decoration={brace}, thick] (ca) -- (ba)
  node[midway, left] (bracket) {$\bigl|\Phi^+\bigr\ket\ $} ;

\end{tikzpicture} \hfill \hbox{}
\caption{The Quantum Teleportation Circuit}
\label{fig:quantum_teleportation}
\end{framed}
\end{figure}
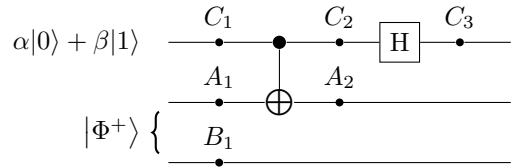

Two particles $A$ and $B$ are prepared in a maximally entangled state~$|\Phi^+\ket$:
$$ (A_1, B_1) \models |\Phi^+\ket \propto |00\ket + |11\ket, $$
and are shared between Alice and Bob. Alice takes particle $A$ and Bob particle~$B$.

Suppose now that Alice has a particle $C$ she wishes to ``teleport'' to Bob, and suppose moreover that
$$ C_1 \models \alpha |0\ket + \beta |1\ket.$$
We thus have $ (A_1, B_1) \models |00\ket + |11\ket$ and $ C_1 \models \alpha |0\ket + \beta |1\ket$ or, equivalently
\begin{multline*}
(A_1, B_1, C_1) \models \bigl(|00\ket + |11\ket \bigr) \otimes \bigl(\alpha |0\ket + \beta |1\ket \bigr) \\ \propto \alpha |000\ket + \beta |001\ket + \alpha |110\ket + \beta |111\ket
\end{multline*}
Alice can proceed as follow:
\begin{enumerate}
  \item she applies a controlled-not operator $\oplus$ to particles $A$ and $C$:
  $$ (A_2, C_2) = \Uni\bigl((A_1, C_1), \oplus\bigr), \quad \hbox{so~that} $$
\begin{multline*}
(A_2, B_1, C_2) \models \alpha |000\ket + \beta |101\ket + \alpha |110\ket + \beta |011\ket \\
\propto \alpha \bigl(|00\ket + |11\ket \bigr) \otimes |0\ket + \beta \bigl(|10\ket + |01\ket \bigr) \otimes |1\ket;
\end{multline*}
\item she then applies a Hadamard gate to particle $C$ (i.e.~$C_3 = \Uni(C_2, H)$), leading~to
$$ (A_2, B_1, C_3) \models |\Psi\ket, \quad \hbox{where} $$
\begin{align*}
|\Psi\ket & \propto \alpha \bigl(|00\ket + |11\ket \bigr) \otimes \bigl(|0\ket + |1\ket \bigr) + \beta \bigl(|10\ket + |01\ket \bigr) \otimes \bigl(|0\ket - |1\ket\bigr) \\
& = \alpha \bigl(|000\ket + |001\ket + |110\ket + |111\ket \bigr) + \beta \bigl(|100\ket - |101\ket + |010\ket - |011\ket \bigr);
\end{align*}
\end{enumerate}
\begin{enumerate} \setcounter{enumi}{2}
\item she measures both her particles with observable $\coll{|0\ket, |1\ket}$. 
\end{enumerate}
Suppose, for instance, that $(\blank, |1\ket) = \Mes(C_3, \blank)$ and $(\blank, |0\ket) = \Mes(A_2, \blank)$. We can deduce that
$$ B_1 \models \appl_{A,C} \bigl( \mathopen\vert01\ket, |\Psi\ket \bigr) \propto \alpha |0\ket - \beta |1\ket .$$

Then, applying unitary operator $\sigma_z$ to $B$ turns it into the ``teleported'' version of particle $C$. Formally, for $B_2 = \Uni(B_1, \sigma_z)$, one~has
$$ B_2 \models \alpha |0\ket + \beta |1\ket, $$ 
so that $B_2$ is now verifies the judgement originally verified by $C_1$.

\ 

This calculation can be generalized to the situation where particle $C$ was initially entangled with another system $D$, with the composite system such~that
$$ (C_1, D) \models \alpha |0\ket \otimes |\varphi_0\ket + \beta |1\ket \otimes |\varphi_1\ket.$$
In that case, one would have obtained
$$ (B_2, D) \models \alpha |0\ket \otimes |\varphi_0\ket + \beta |1\ket \otimes |\varphi_1\ket.$$

In this approach, the term ``teleportation'' might appear as a bit excessive, since all that is expressed here is that if a particle $C$ is \emph{known} at the beginning to verify some $|\varphi\ket$, then after the ``teleportation'' has occured, particule $B$ is \emph{known} to verify the same $|\varphi\ket$. We do indeed insist on the fact that the relation between verification statements is an one-way implication and not an equivalence:
$$ C_1 \models |\varphi\ket \implies B_2 \models |\varphi\ket $$
or, more generally
$$ C_1, D \models |\psi\ket \implies B_2, D \models |\psi\ket.$$ 
Thus, we can only assert that any verification judgement that could be made concerning $C_1$ can now be made concerning $B_2$. Operationally however, regarding measurement outcomes, this means that $B_2$ behaves exactly the same way as $C_1$ would have before the teleportation.

\subsection{Measuring an EPR pair}

Let us now focus on two spin-$\frac 1 2$ particles $A$ and $B$ prepared in maximally entangled Bell's state $|\Psi^-\ket$, that~is
  \[
    (A, B) \models |\Psi^-\ket \propto |0, 1\ket - |1, 0\ket.
  \]
If we define
\[
   | u(\theta, \varphi) \ket = \cos \theta |0\ket + \sin \theta e^{i\varphi} |1\ket \quad \hbox{and} \quad | v(\theta, \varphi) \ket = - \sin \theta |0\ket + \cos \theta e^{i\varphi} |1\ket,
 \]
we have
\[
  \bigl| \appl_1 (u(\theta, \varphi), \Psi^-) \bigr\ket = \cos \theta |1\ket - \sin \theta e^{-i\varphi} |0\ket \propto |v(\theta, \varphi)\ket.
\]
In order to interpret this result, suppose again that particle $A$ (resp. particle $B$) is given to Alice (resp. Bob), and that Alice measures her particle with outcome of~$|u(\theta, \varphi)\ket$:
$$ \bigl(\blank, |u(\theta, \varphi)\ket \bigr) = \Mes (A, \blank). $$

With the knowledge of both $(A,B) \models |\Psi^-\ket$ and $(\blank, |u(\theta, \varphi\ket) = \Mes(A, \blank)$, one can infer that $B \models |v(\theta, \varphi)\ket$. Typically, Alice will be able to make such an inference after having measured her particle. But the original judgement, namely $(A,B) \models |\Psi^-\ket$, remains valid since, from its definition as an implication regarding the potential outcome of a measurement, its validity is independant of whether the considered system has actually been measured.



We thus now have two valid judgements about $B$:
$$ (A, B) \models |\Psi^-\ket, \qquad B \models |v(\theta, \varphi)\ket. $$

Obviously, the choice of applying $\sigma_z$ (and more generally of which unitary operator should be applied to $B$) follows from the outcomes obtained by Alice.

\ 


Now, suppose that Bob measures his particle, with the same observable $\coll{|u(\theta, \varphi)\ket, |v(\theta, \varphi)\ket}$. Since $B \models |v(\theta, \varphi)\ket$, outcome $|v(\theta, \varphi)\ket$ is certain:
$$ \bigl( \blank, |v(\theta, \varphi)\ket \bigr) = \Mes (B, \coll{|u(\theta, \varphi)\ket, |v(\theta, \varphi)\ket}). $$

But then, knowing the outcome of the measurement of $B$, since $(A, B) \models |\Psi^-\ket$ is still valid, it is possible to deduce~that
$$ A \models \mathopen| \appl_2(v(\theta, \varphi), \Psi^-) \ket \propto | u(\theta, \varphi) \ket. $$

Having now both $A \models | u(\theta, \varphi) \ket$ and $B \models |v(\theta, \varphi)\ket $, we indeed have
$$ (A, B) \models | u(\theta, \varphi) \otimes v(\theta, \varphi)\ket. $$
We thus have obtained two correct judgements regarding $A$ and $B$:
$$  (A, B) \models \mathopen|\Psi^-\ket \quad \hbox{and} \quad (A, B) \models | u(\theta, \varphi) \otimes v(\theta, \varphi)\ket. $$


This illustrates that, as announced earlier, it is possible to that distinct and yet valid judgements about the same quantum system. Moreover, if Alice and Bob exchange their measurement outcomes, they become both able to express the previous two judgements about $A$ and $B$.

This situation can be related to Relational Quantum Mechanics \cite{Rovelli96RQM} from which we quote:
\begin{quotation}
``\textbf{Main Observation:} In quantum mechanics different observers may give different accounts of the same sequence of events.'',
\end{quotation}
but in QML, several accounts of the same situation can be obtained and, these accounts being made using verification judgement which express classical knowledge about a quantum system, they are valid independently of who, where and when they are obtained and can, to that respect, be considered as objectively true.

This is, we believe, a major gain of switching from a state vector-based description to a verification judgement-based one: everything that can be expressed in QML represents purely epistemic and observer-independant statements about a quantum system.





\subsection{Hardy's Paradox}

We now analyze Hardy's Paradox \cite{Hardy92,Hardy93,Hardy2000QuantumCakes} in the light of verification judgements. Following the presentation of this paradox given by Mermin \cite{Mermin94QuantumMysteriesRefined,Mermin98Ithaca,Mermin98What}, we consider two qubits $A$ and $B$, and for each qubit, two observables with eigenvectors $|0\ket$ and $|1\ket$ on the one hand, and $|+\ket$ and $|-\ket$ on the other hand. We suppose that system $(A,B)$ has been prepared in state $|\Psi_{\!H}\ket$ defined~as
$$ |\Psi_{\!H}\ket = |+_A, +_B\ket - \bra 1_A, 1_B | +_A, +_B \ket |1_A, 1_B\ket. $$
Using the partial application operation, we~have
\begin{align*}
  \bigl| \appl_2(-_B, \Psi_{\!H}) \bigr\ket & = \bra -_B | +_B \ket | +_A \ket - \bra -_B | 1_B \ket \bra 1_A, 1_B | +_A, +_B \ket | 1_A \ket \\
  & = - \bra -_B | 1_B \ket \bra 1_A, 1_B | +_A, +_B \ket | 1_A \ket \\
  & \propto | 1_A \ket,
\end{align*}
\begin{align*}
  \bigl| \appl_1(1_A, \Psi_{\!H}) \bigr\ket & = \bra 1_A | +_A \ket | +_B \ket - \bra 1_A | 1_A \ket \bra 1_A, 1_B | +_A, +_B \ket | 1_B \ket \\
  & \propto | +_B \ket - \bra 1_B | +_B \ket | 1_B \ket \\
  & \propto | 0_B \ket,
\end{align*}
since $|+_B\ket = \bra 0_B | +_B \ket |0_B\ket + \bra 1_B | +_B \ket |1_B\ket $, and~finally
%
%
\begin{align*}
  \bigl| \appl_2(0_B, \Psi_{\!H}) \bigr\ket & = \bra 0_B | +_B \ket | +_A \ket - \bra 0_B | 1_B \ket \bra 1_A, 1_B | +_A, +_B \ket | 1_A \ket \\
  & = \bra 0_B | +_B \ket |+_A \ket \\
  & \propto | +_A \ket.
\end{align*}
These calculations can be ``chained'', as represented in figure~\ref{fig_hardy}, and constitute the core of Hardy's Paradox, since obtaining outcome $|-_B\ket$ when measuring particle $B$ seems to imply that $|+_A\ket$ should be a certain outcome for $A$ (so that outcomes $|-_A\ket$ and $|-_B\ket$ cannot be obtained simultaneously). Let us focus now on the way this can be interpreted using QML, and see whether this paradox remains.

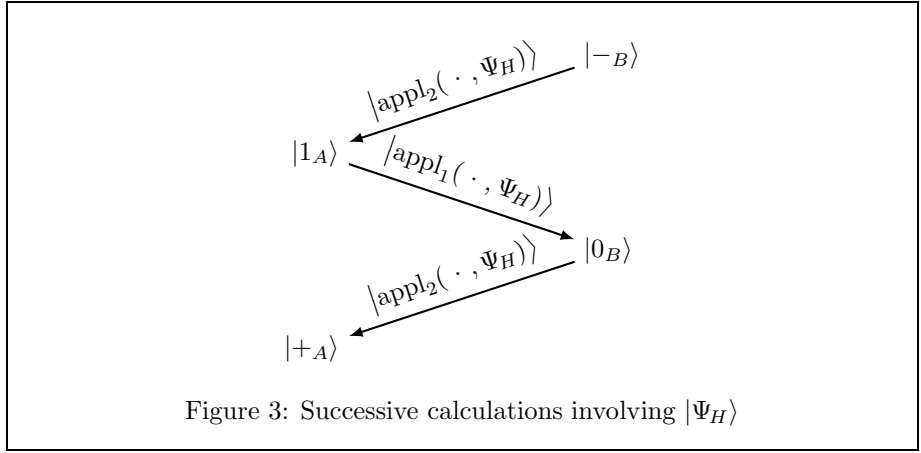
\begin{figure}[htbc]
  \begin{framed}
\centering
\begin{tikzpicture}[baseline=(current bounding box.center),
    node distance = 7mm and 30mm]
\node (a) {$|-_B\ket$} ;
\node (b) [below left = of a] {$|1_A\ket$} ;
\node (c) [below right = of b] {$|0_B\ket$} ;
\node (d) [below left = of c] {$|+_A\ket$} ;
\draw[thick,-latex] (a) -- (b) node[midway, sloped, above] {$\bigl| \appl_2(\ \cdot \ ,\Psi_{\!H}) \bigr\ket$};
\draw[thick,-latex] (b) -- (c) node[midway, sloped, above] {$\bigl| \appl_1(\ \cdot \ ,\Psi_{\!H}) \bigr\ket$};
\draw[thick,-latex] (c) -- (d) node[midway, sloped, above] {$\bigl| \appl_2(\ \cdot \ ,\Psi_{\!H}) \bigr\ket$};
\end{tikzpicture}
\caption{Successive calculations involving $|\Psi_{\!H}\ket$}
\label{fig_hardy}
\end{framed}
\end{figure}

\ 

We suppose that system $(A, B)$ verifies $|\Psi_{\!H} \ket $ and that neither $A$ nor $B$ have been measured yet. If Bob measures $B$ and obtains outcome $|-_B\ket$, this implies, since $\mathopen| \appl_2(-_B, \Psi_{\!H}) \ket \propto |1_A\ket$, that $A \models |1_A\ket$:
$$ (\blank, |-_B\ket) = \Mes(B, \blank) \implies A \models |1_A\ket $$

If now Alice measures $A$ with observable $\bigl\{|0_A\ket, |1_A\ket\bigr\}$, then of course she is to obtain outcome $|1_A\ket$:
$$
\bigl[ A \models |1_A\ket \cand (\blank, |\varphi\ket) = \Mes(A, \{|0_A\ket, |1_A\ket\})\bigr] \implies |\varphi\ket = |1_A\ket $$
As a consequence, we deduce that $\mathopen| \appl_1(1_A, \Psi_{\!H})\ket \propto |0_B\ket$ is a certain outcome for~$B$:
$$ (\blank, |1_A\ket) = \Mes(A, \blank) \implies B \models |0_B\ket $$

But at this point, $B$ has already been measured by Bob with outcome $|-_B\ket$, so that it is not possible to have $|0_B\ket$ as an outcome when measuring $B$, and hence, the deduction step:
$$ (\blank, |0_B\ket) = \Mes(B, \blank) \implies A \models \mathopen| \appl_2(0_B, \Psi_{\!H}) \mathclose\ket \propto | +_A \ket $$
cannot be made.

\ 

As a consequence, using QML and by subjecting the order of measurement to a careful scrutiny, it seems that it is not possible to deduce that $A$ verifies $|+_A\ket$ if outcome $|-_B\ket$ has been obtained when measuring~$B$, thus nullifying Hardy's Paradox.

\section{Perspectives} 

Let us first mention an element which is usually considered as consitutive of quantum mechanics, and which has not been evoked in our approach, namely probabilities: the Born Rule is inherently a rule telling, given the state vector of the system being measured, the probabilities of obtaining a given outcome. Since in our approach, state vectors where replaced by verification judgements, it would seem natural to adapt the Born Rule in the following way:
$$ \mathrm{prob}\Bigl((\blank, |\varphi\ket) = \Mes(S, \blank) \Big\vert S \models |\psi\ket\Bigr) = \bigl\vert \bra \psi \vert \varphi \ket \bigr\vert^2 $$
However, as we have seen, it is possible that different verification judgements describe the same situation. For instance, in the case of an EPR pair made of two particles $A$ and $B$, one could have, simultaneously
$$ (A,B) \models |\Psi^-\ket \quad \hbox{and} \quad (A,B) \models |0 \otimes 1\ket.$$
As a consequence, such a direct adaption of the Born Rule is not possible without a deeper understanding of the meaning of verification judgements. 

\ 

Another point which needs to be developed is the rigourous representation of the experimental setup and, more generally, the arrangement of the different components (unitary operations and measurements) acting of a system. As shown in the study of Hardy's paradox, it is important to organize things sufficiently so that a given part of the system can be acted on only once (in order to avoid, for instance, that a single element be measured twice).

\

Finally, we have only presented here a simplified version of the QML formalism, with only vector rays and nond-degenerate observables. A more general version of this formalism should be able to take into account every closed subspace of a Hilbert space, and every observable.

\bibliographystyle{apalike}

\end{document}